\documentclass[a4paper,11pt]{article}

\usepackage{fullpage}
\usepackage{times}
\usepackage{soul}
\usepackage{url}
\usepackage[utf8]{inputenc}
\usepackage[small]{caption}
\usepackage{graphicx}
\usepackage{amsmath}
\usepackage{booktabs}

\usepackage[ruled,vlined,linesnumbered]{algorithm2e}

\usepackage{enumitem}

\usepackage{libertine}

\usepackage{amssymb,amsmath,amsfonts,amstext,amsthm}

\usepackage{hyperref}
\usepackage[svgnames]{xcolor}
\usepackage[capitalise,nameinlink]{cleveref}
\hypersetup{colorlinks={true},linkcolor={DarkBlue},citecolor=[named]{DarkGreen}}

\usepackage{bbm}

\newtheorem{theorem}{Theorem}[section]
\newtheorem{corollary}[theorem]{Corollary}
\newtheorem{proposition}[theorem]{Proposition}
\newtheorem{lemma}[theorem]{Lemma}

\theoremstyle{definition}
\newtheorem{definition}[theorem]{Definition}
\newtheorem{example}[theorem]{Example}

\usepackage{natbib}

\usepackage{mathtools}

\usepackage{authblk}

\usepackage{subcaption}

\usepackage{pgfplots}
\pgfplotsset{compat=newest}
\usepgfplotslibrary{ternary}
\usepackage{tikz}

\newcommand{\BR}{\textit{{BR}}}
\newcommand{\fBR}{\widetilde{\BR}}
\newcommand{\uff}{\tilde{u}^F}

\newcommand{\xx}{\mathbf{x}}
\newcommand{\yy}{\mathbf{y}}
\newcommand{\zz}{\mathbf{z}}
\newcommand{\ww}{\mathbf{w}}

\newcommand{\gv}{\mathbf{g}}
\newcommand{\ee}{\mathbf{e}}
\newcommand{\bb}{\mathbf{b}}
\newcommand{\hh}{\mathbf{h}}

\newcommand{\calG}{\mathcal{G}}

\DeclareMathOperator{\relint}{relint}

\DeclareMathOperator*{\argmin}{argmin}
\DeclareMathOperator*{\argmax}{argmax}
\renewcommand{\emptyset}{\varnothing}

\begin{document}

\allowdisplaybreaks

\title{\bf Optimally Deceiving a Learning Leader \\ in Stackelberg Games\thanks{Georgios Birmpas is supported by the ERC Starting grant number 639945 (ACCORD). Jiarui Gan is supported by the EPSRC International Doctoral Scholars Grant EP/N509711/1. Alexandros Hollender is supported by an EPSRC doctoral studentship (Reference 1892947).
}}

\author[1]{Georgios Birmpas}
\author[1]{Jiarui Gan}
\author[1]{Alexandros Hollender}
\author[1]{\\ Francisco J. Marmolejo-Coss{\'\i}o}
\author[1]{Ninad Rajgopal}
\author[2]{Alexandros A. Voudouris}

\affil[1]{University of Oxford, UK}
\affil[2]{University of Essex, UK}

\date{}

\maketitle   

\begin{abstract}
Recent results in the ML community have revealed that learning algorithms used to compute the optimal strategy for the leader to commit to in a Stackelberg game, are susceptible to manipulation by the follower. Such a learning algorithm operates by querying the best responses or the payoffs of the follower, who consequently can deceive the algorithm by responding as if his payoffs were much different than what they actually are. For this strategic behavior to be successful, the main challenge faced by the follower is to pinpoint the payoffs that would make the learning algorithm compute a commitment so that best responding to it maximizes the follower's utility, according to his true payoffs. While this problem has been considered before, the related literature only focused on the simplified scenario in which the payoff space is finite, thus leaving the general version of the problem unanswered. In this paper, we fill in this gap, by showing that it is always possible for the follower to compute (near-)optimal payoffs for various scenarios about the learning interaction between leader and follower.
\end{abstract}    

\section{Introduction}
{\em Stackelberg games} are a simple yet powerful model for sequential interaction among strategic agents. In such games there are two players: a leader and a follower. The leader commits to an action, and the follower acts upon observing the leader's commitment. The simple sequential structure of the game permits modeling a multitude of important scenarios. Indicative applications include the competition between a large and a small firm \citep{von2010market}, the allocation of defensive resources \citep{tambe2011security}, the competition among mining pools in the Bitcoin network~\citep{marmolejo2019competing,sun2020miners}, and the protection again manipulation in elections~\citep{elkind2019protecting,yin2018optimal}.

In Stackelberg games, the leader is interested in finding the best commitment she can make, assuming that the follower behaves rationally. The combination of such a commitment by the leader and the follower's rational best response to it leads to a strong Stackelberg equilibrium (SSE). In general, the utility that the leader obtains in an SSE is larger than what she would obtain in a Nash equilibrium of the corresponding one-shot game~\citep{von2004leadership}, implying that the leader prefers to commit than to engage in a simultaneous game with the follower. 

In case the leader has access to both hers and the follower's payoff parameters, computing an SSE is a computationally tractable problem~\citep{conitzer2006computing}. In practice however, the leader may have limited or no information about the follower's payoffs. Consequently, in order to determine the optimal commitment, the leader must endeavor to elicit information about the incentives of the follower through indirect means. This avenue of research has led to a plethora of active-learning-based approaches for the computation of SSEs~\citep{balcan2015committment,blum2015learning,letchford2009learning,peng2019learning,roth2016watch}. 
At the same time, inspired by recent developments in the ML community regarding adversarial examples in classification algorithms \citep{barreno2010security,lowd2005adversarial},
there has been a stream of recent papers exploring the notion of adversarial deception by  the follower, when facing algorithms used by the leader for learning SSEs in Stackelberg games. 

Specifically, when an algorithm learns an SSE by querying the follower's best responses, the follower can use fake best responses to distort the SSE learned by the algorithm. As recently explored by \citet{gan2019imitative}, one particular approach the follower can employ, is to imitate best responses implied by payoffs that are different from his actual ones. Therefore, the key to the success of such a deceptive behavior is to pinpoint the fake payoffs that could make the leader learn an SSE in which the actual utility of the follower is maximized. In the scenario studied in \citep{gan2019imitative}, this task is trivial as the follower's choices are limited to a finite set of polynomially many payoff matrices; thus, to efficiently find out the optimal payoffs, the follower can simply enumerate all possible matrices.

To the best of our knowledge, the general version of this problem, where the follower is allowed to use {\em any} payoff matrix, without restrictions on the space of possible values, has been considered only in two very recent papers~\citep{gan2019manipulating,nguyen2019imitative}, which however focused on the specific application of Stackelberg games to security resource allocation problems. Besides that, no progress has been made for general Stackelberg games. In this paper, we aim to fill in this gap, by completely resolving this computational problem, a result that reflects the insecurity of learning to commit in Stackelberg games. 

\subsection*{Our Contribution}
We explore how a follower can optimally deceive a learning leader in Stackelberg games by misreporting his payoff matrix, and study the tractability of the corresponding optimization problem. As in previous work, our objective is to compute the fake payoff matrix according to which the follower can best respond to make the leader learn an SSE in which the true utility of the follower is maximized. However, unlike the related literature, we do not impose any restrictions on the space from which the payoffs are selected or on the type of the game. 
By exploiting an intuitive characterization of all strategy profiles that can be induced as SSEs in Stackelberg games, we show that it is always possible for the follower to compute a payoff matrix implying an SSE which maximizes his true utility, in polynomial time. Furthermore, we strengthen this result to resolve possible equilibrium selection issues, by showing that the follower can construct a payoff matrix that induces a {\em unique} SSE, in which his utility is maximized up to some arbitrarily small loss.

\subsection*{Other Related Work}
Our paper is related to an emerging line of work at the intersection of machine learning and algorithmic game theory, dealing with scenarios where the samples used for training learning algorithms are controlled by strategic agents, who aim to optimize their personal benefit. Indicatively, there has been recent interest in the analysis of the effect of strategic behavior on the efficiency of existing algorithms, as well as the design of algorithms resilient to strategic manipulation for linear regression~\citep{ben-porat2019regression,chen2018regression,dekel2010regression,hossain2020regression,perote2004regression,waggoner2015market} and classification~\citep{chen2019grinding,dong2018classification,meir2012classification,shang2019samples}.
 
Beyond the strategic considerations above, our work is also related to the study of query protocols for learning game-theoretic equilibria. In this setting, as in ours, algorithms for computing equilibria via utility and best response queries are a natural starting point. For utility queries, there has been much work in proving exponential lower bounds for randomized computation of exact, approximate and well-supported Nash equilibria~\citep{BabR,B13,CCT15,GR14,HM10,HN13}, as well as providing query-efficient protocols for approximate Nash equilibrium computation in bimatrix games, congestion games \citep{FGGS}, anonymous games \citep{GT14}, and large games \citep{GMZ19logarithmic}. 
Best response queries are forcibly weaker than utility queries, but they arise naturally in practice, and are also expressive enough to implement fictitious play, a dynamic first proposed in \cite{Brown}, and proven to converge in \citep{Robinson} for two-player zero-sum games to an approximate Nash equilibrium. In terms of equilibrium computation, the authors in \citep{goldberg2018learning} also provide query-efficient algorithms for computing approximate Nash equilibria for bimatrix games via best response queries provided one agent has a constant number of strategies.

Finally, learning via incentive queries in games is directly related to the theory of preference elicitation, where the goal is to mine information about the private parameters of the agents by interacting with them~\citep{BJSZ04,LahaieParkes04,zinkevich2003polynomial,GLM20}. This has many applications, most notably combinatorial auctions, where access to the valuation functions of the agents is achieved via value or demand queries~\citep{combinatorial-chapter,ConenSandholm01,NisanSegal06}. 

\section{Preliminaries}\label{sec:prelim}
A Stackelberg game (SG) is a sequential game between a \emph{leader} and a \emph{follower}\footnote{Following the standard convention, we will refer to the leader as a female and to the follower as a male.}. 
The leader commits to a strategy, and the follower then acts upon observing this commitment. We consider finite SGs, in which the leader and the follower have $m$ and $n$ {\em pure strategies} at their disposal, respectively, and their utilities for all possible outcomes are given by the matrices $u^L,u^F \in \mathbb{R}^{m \times n}$. The entries $u^L (i,j)$ and $u^F(i,j)$ denote the utilities of the leader and the follower, under {\em pure strategy profile} $(i,j) \in [m] \times [n]$. We use $\calG = (u^L,u^F)$ to denote the SG with payoff matrices $u^L$ and $u^F$; we omit $m$ and $n$ as they are clear from context.

Like one-shot games, the agents are allowed to employ mixed strategies whereby they randomize over actions in their strategy set. A mixed strategy of the leader is a probability distribution over $[m]$, denoted by $\xx \in \Delta^{m-1} = \{\xx \geq 0 : \sum_{i \in [m]} x_i = 1 \}$. 
By slightly abusing notation, we let 
$u^L(\xx,j) = \sum_{i \in [m]} x_i \cdot u^L(i,j)$
be the {\em expected utility} of the leader when she plays the mixed strategy $\xx$ and the follower plays a pure strategy $j$.
Similarly, we define $u^F(\xx,j) = \sum_{i \in [m]} x_i \cdot u^F(i,j)$ for the follower. For a given mixed strategy $\xx \in \Delta^{m-1}$ of the leader, we say that $j \in [n]$ is a {\em follower best response} if $u^F(\xx,j) = \max_{\ell \in [n]} u^F(\xx,\ell)$; we denote the set of all follower best responses to $\xx$ by $\BR(\xx) \subseteq [n]$ and refer to the function $\BR$ as the {\em best response correspondence} of the follower. 

A {\em strong Stackelberg equilibrium} (SSE) is the standard solution concept in SGs, and captures the situation where the leader commits to a mixed strategy that maximizes her expected utility, while taking into account the follower's best response to her commitment. It is assumed that the follower breaks ties in favor of the leader when he has multiple best responses.
\footnote{This standard assumption is justified by the fact that such tie-breaking behavior can often be enforced by an infinitesimal perturbation in the leader's strategy \citep{von2004leadership}.} 

\begin{definition}[SSE]
\label{def:sse}
A strategy profile $(\xx,j)$ is an SSE of the SG $\calG = (u^L,u^F)$ if
\begin{align*}
	( \xx, j ) \in {\arg\max}_{\yy \in \Delta^{m-1}, \, \ell \in \BR(\yy) } u^L (\yy, \ell ).
\end{align*}
\end{definition}

\paragraph{Learning SSEs and Deceptive Follower Behavior.}
We consider the scenario where the leader has full knowledge of her utility matrix $u^L$, and aims to compute an SSE by interacting with the follower and gleaning information about $u^F$. For example, the leader could observe follower best responses in play (akin to having query access to $\BR$), or observe follower payoffs at pure strategy profiles during play (akin to having query access to $u^F$ as a function). Hence, this can be cast as the problem of learning an SSE with a specified notion of query access to information about the follower's incentives. 

Consider an SG $\calG = (u^L,u^F)$. If the follower controls the flow of information to the leader in this paradigm, he may consider perpetually interacting with the leader as if he had a different payoff matrix $\uff$, which can make the leader believe that both agents are playing the game $\widetilde{\calG}=(u^L,\uff)$. This deceiving power provides the follower with an incentive to act according to $\widetilde{\calG}$ for a judicious choice of $\uff$, because the SSEs in $\widetilde{\calG}$ may provide larger utility (according to $u^F$) than the SSEs in $\calG$. More concretely, the example below shows that the follower can gain an arbitrary benefit by deceiving the leader to play a different game.

\begin{example}[\bf Beneficial deception]
Let $\alpha \in [0,1]$ and consider the following matrices:
\[
R = 
\begin{pmatrix}
1 & 0 \\
0 & 0
\end{pmatrix},\quad
C_\alpha= 
\begin{pmatrix}
0 & \alpha  \\
1 & \alpha 
\end{pmatrix}
\]
Now, suppose that $u^L = R$ and $u^F = C_\alpha$, and let $x \in [0,1]$ represent the probability mass that the leader (row player) places on the first row (her first strategy); thus, $1-x$ is the probability with which she plays her second strategy. Given this mixed strategy of the leader, the utilities that the follower expects to derive from her two strategies (columns) are $u^F(x,1) = 1-x$ and $u^F(x,2)=\alpha$. Consequently, the first strategy is a best response of the follower when $x \in [0,1-\alpha]$, and the second one is a best response when $x \in (1-\alpha,1]$ (when $x = 1 - \alpha$, the tie is broken in favor of the leader). With this information, it is clear that the SSE of the game occurs when the leader chooses $x = 1-\alpha$ and the follower plays his first strategy. As a result, the follower's utility is $u^F(1-\alpha,1)=\alpha$. 

However, for any $\alpha < 1$, the follower has an incentive to deceive the leader into playing the game $\widetilde{\calG} = (R, C_1)$, which will improve his utility in the resulting SSE to $1$. This will be an improvement by a multiplicative factor of $1/\alpha$, which can be arbitrarily large when $\alpha$ is arbitrarily close to $0$.
\hfill $\qed$
\end{example}

\paragraph{Inducible Strategy Profiles.}
The ultimate goal of the follower is to identify the SSE that maximizes his true utility, from the set of SSEs that he can deceive the leader into learning. We will refer to such SSEs as {\em inducible strategy profiles}. At a high level, the follower's problem can now be expressed as the following optimization problem:
\begin{align} 
\label{prg:high-level}
&\max_{\xx, j}  \quad u^F(\xx, j), \\
&\text{ subject to } \quad  (\xx,j) \text{ is inducible} \nonumber
\end{align}
This maximum utility for the follower is called the \emph{optimal inducible utility}. If the maximum value is never achieved, then for every $\varepsilon > 0$, we would like to be able to find an inducible SSE that achieves a value $\varepsilon$-close to the supremum value.

As discussed previously, the leader can learn an SSE by gleaning information about the incentives of the follower by querying the best responses of the follower to particular leader strategies, or more refined information about the follower's payoff matrix. Depending on the type of information queried, we can define various levels of inducible strategy profiles.

In more detail, suppose the leader can only query the best responses of the follower, who behaves according to some best response correspondence $\fBR: \Delta^{m-1} \rightarrow 2^{[n]} \setminus \{\emptyset\}$. This interaction between the leader and the follower leads to a game $\widetilde{\calG} = (u^L, \fBR)$ where only information about $\fBR$ is known (instead of a payoff matrix implying $\fBR$). The definition of $\fBR$ enforces a best response answer to any possible query. Consequently, the leader learns an SSE $( \xx, j ) \in {\arg\max}_{\yy \in \Delta^{m-1}, \, \ell \in \fBR(\yy) } u^L (\yy, \ell )$, which yields the following notion of {\em BR-inducible} strategy profiles.

\begin{definition}[BR-inducibility]\label{def:BR-inducible}
A strategy profile $(\xx, j)$ is {\em BR-inducible} with respect to $u^L$ if there exists a best response correspondence $\fBR: \Delta^{m-1} \rightarrow 2^{[n]} \setminus \{\emptyset\}$ such that $(\xx,j)$ is an SSE of the game $\widetilde{\calG} = (u^L, \fBR)$, in which case we say that $(\xx, j)$ is induced by $\fBR$.
\end{definition}

Next, consider the case where the leader can query information about the payoffs of the follower, who can now behave according to a fake payoff matrix $\uff$. We refer to the SSEs of the resulting game $\widetilde{\calG} = (u^L,\uff)$ as {\em payoff-inducible} strategy profiles. 

\begin{definition}[Payoff-inducibility]\label{def:payoff-inducible}
A strategy profile $(\xx, j)$ is said to be {\em payoff-inducible} with respect to $u^L$ if there exists $\uff \in \mathbb{R}^{m\times n}$ such that $(\xx, j)$ is an SSE in the game $\widetilde{\calG} = (u^L,\uff)$, in which case we say that $(\xx, j)$ is induced by $\uff$.
\end{definition}

Clearly, payoff-inducibility is stricter than BR-inducibility: for every choice of $\uff$, the corresponding best response correspondence $\fBR(\yy) = \arg\max_{\ell\in[n]} \uff(\yy, \ell)$ induces the same SSEs as $\uff$ does. 

Note that the above definitions only require an inducible strategy profile to be a verifiable SSE, with respect to information about the follower's incentive (either $\fBR$ or $\uff$). 
However, it may happen that the resulting game $\widetilde{\calG}$ has multiple SSEs, which gives rise to an equilibrium selection issue. Indeed, in practice, it is not realistic to assume that the follower has any control over which SSE is chosen by the leader (who moves the first in the game). To address this, and thus completely resolve the optimal deception problem for the follower, we introduce an even stricter notion of inducibility on top of payoff-inducibility, which requires $\widetilde{\calG}$ to have a unique SSE.

\begin{definition}[Strong inducibility]\label{def:strong-inducible}
A strategy profile $(\xx,j)$ is said to be \emph{strongly inducible} with respect to $u^L$, if there exists a matrix $\uff \in \mathbb{R}^{m\times n}$ such that $(\xx,j)$ is the \emph{unique} SSE of the game $\widetilde{\calG}=(u^L,\uff)$, in which case we say that $(\xx,j)$ is \emph{strongly} induced by $\uff$. 
\end{definition}

In the next sections, we will investigate solutions to \eqref{prg:high-level} under the inducibility notions above, from the weakest to the strongest.
Our general approach is to decompose \eqref{prg:high-level} into $n$ sub-problems by enumerating all possible follower responses $j \in [n]$. For each strategy $j$, we solve the corresponding optimization problem, and pick the one that yields the maximum utility for the follower. Due to space constraints, some proofs are omitted and can be found in the supplementary material. 

\section{Best Response Inducibility}
\label{sec:best-response}
Let us start our analysis by considering the case in which the leader queries the best responses of the follower. The aim of the follower is to deceive the leader towards a strategy profile that is BR-inducible; see Definition~\ref{def:BR-inducible}. Indeed, if the follower is allowed to use an arbitrary $\fBR$ to induce a strategy profile $(\xx,j)$, he can simply define $\fBR$ as follows:
\begin{align*}
\fBR(\yy) =
\begin{cases}
\{j\} & \text{if } \yy = \xx \\
\arg\min_{\ell \in [n]} u^L(\yy, \ell) & \text{if } \yy \neq \xx.
\end{cases}
\end{align*}
Namely, the follower threatens to choose the worst possible response against any leader strategy $\yy \neq \xx$, so as to minimize the leader's incentive to commit to these strategies. 
This $\fBR$ will successfully convince the leader that $(\xx,j)$ is an SSE of $\widetilde{\calG}$, hence inducing $(\xx, j)$, if the threat is powerful enough, that is, if $u^L(\xx, j) \ge \min_{\ell \in [n]} u^L(\yy, \ell)$ for all $\yy \in \Delta^{m-1}$. Equivalently, this means that
\begin{align}
\label{eq:arbitrary-BR-condition}
u^L(\xx, j) \ge M:= \max_{\yy\in\Delta^{m-1}} \min_{\ell \in [n]} u^L(\yy, \ell),
\end{align}
where $M$ is exactly the leader's {\em maximin utility}.
Indeed, \eqref{eq:arbitrary-BR-condition} is necessary for $(\xx,j)$ to be BR-inducible: if on the contrary $u^L(\xx, j) < M$, then by committing to $\yy^* \in \arg\max_{\yy\in\Delta^{m-1}} \min_{\ell \in [n]} u^L(\yy, \ell)$, the leader can obtain (at least) her maximin utility, which will be strictly larger than $u^L(\xx,j)$. 

Thus, condition \eqref{eq:arbitrary-BR-condition} gives a simple criterion for BR-inducibility.
The problem is that such $\fBR$ may be far from being one that arises from a choice of $\uff$. To alleviate this limitation, we impose a stricter condition on $\fBR$.

\paragraph{Polytopal BR Correspondence.}
In a similar vein to \citet{goldberg2018learning}, we require that, for every $\ell \in [n]$, the set of leader strategies to which $\ell$ is a best response $\fBR^{-1}(\ell) = \{\yy \in \Delta^{m-1} : \ell \in \fBR(\yy) \}$ is a closed convex polytope, and the union of all these sets forms a partition of $\Delta^{m-1}$ (for example, see the polytope partition of $\Delta^2$ in Figure~\ref{fig:no-payoff-matix}).
Any best response correspondence $\fBR$ satisfying this assumption is called {\em polytopal}.

\begin{definition}[Polytopal best response correspondence \citep{goldberg2018learning}] \label{def:pp-best-response}
A best response correspondence $\fBR: \Delta^{m-1} \rightarrow 2^{[n]} \setminus \{\emptyset\}$ is {\em polytopal} if it also satisfies the following:
\begin{itemize}
    \item $\fBR^{-1}(\ell)$ is a closed convex polytope for each $\ell \in [n]$, and
    \item For each $k \neq \ell$, either $\relint(\fBR^{-1}(k)) \cap \relint(\fBR^{-1}(\ell)) = \emptyset$ or $\fBR^{-1}(k) = \fBR^{-1}(\ell)$, where $\relint(H)$ denotes the relative interior of a set $H$.
\end{itemize}
\end{definition}

Being polytopal is necessary for $\fBR$ to arise from some payoff matrix. Indeed, the {\em true} best response correspondence $\BR$ that arises from $u^F$ is polytopal:
clearly, each $\BR^{-1}(\ell)$ is a closed convex polytope defined by the hyperplanes $u^F(\yy, \ell) \ge u^F(\yy, k)$ for all $k \in [n]$ and the borders of $\Delta^{m-1}$; in addition, $\cup_{\ell=1}^n \BR^{-1}(\ell) = \Delta^{m-1}$, and for any $\ell \neq k$, the polytopes $\BR^{-1}(\ell)$ and $\BR^{-1}(k)$ only intersect at their borders unless $u^F(\cdot,\ell) = u^F(\cdot,k)$.
Thus, if the follower attempts to deceive the leader via a fake $\fBR$, the leader might spot the deception in case $\fBR$ is not polytopal.

It turns out that the following correspondence, which we denote as $\fBR_{\text{P}}$, is polytopal and, as we will shortly show, it is in fact as powerful as any best response correspondence.
\begin{align*}
\fBR_{\text{P}}(\yy) =
\begin{cases}
\{j\} & \text{if } 
\yy \in \Delta^{m-1} \setminus \overline{U_j(\xx)} \\[1mm]
\{j\} \cup \arg\min_{\ell \in [n]\setminus \{j\}} u^L(\yy, \ell) & \text{if } \yy \in \overline{U_j(\xx)} \setminus U_j(\xx) \\[1mm]
\arg\min_{\ell \in [n]\setminus \{j\}} u^L(\yy, \ell) & \text{if } 
\yy \in U_j(\xx)
\end{cases}
\end{align*}
where $\overline{U_j(\xx)}$ is the closure of $U_j(\xx) = \left\{\yy \in \Delta^{m-1} : u^L(\yy, j) > u^L(\xx, j) \right\}$.\footnote{Note that the use of $\overline{U_j(\xx)}$, instead of the set $\left\{\yy \in \Delta^{m-1} : u^L(\yy, j) \ge u^L(\xx, j) \right\}$, is important: when $u^L(\yy,j) = u^L(\xx,j)$ for all $\yy \in \Delta^{m-1}$, these two sets define different behaviors.}
Intuitively, it is safe for the follower to respond by playing $j$ against any leader strategy $\yy$ if $u^L(\yy, j) \le u^L(\xx, j)$, in which case the leader does not have a strong incentive to commit to $\yy$ instead of $\xx$. In response to the other strategies, however, the follower needs to play a different strategy in order to minimize the leader's incentive to commit to such a $\yy$.
Therefore, this approach will successfully induce $(\xx, j)$ if and only if the following holds:
\begin{align}
\label{eq:polytopal-BR-condition}
u^L(\xx, j) \ge \max_{\yy\in \overline{U_j(\xx)}} \min_{\ell \in [n]\setminus\{j\}} u^L(\yy, \ell),
\end{align}
where we use the convention that $\max \emptyset = - \infty$. It is easy to see that $\fBR_{\text{P}}$ is indeed polytopal: 
$\fBR_{\text{P}}^{-1}(j) = \Delta^{m-1} \setminus U_j(\xx)$ is a closed convex polytope, and the same holds for the sets
$\fBR_{\text{P}}^{-1}(\ell)$ defined by the hyperplanes $u^L(\yy, \ell) \le u^L(\yy, k)$, $k\in[n]\setminus\{j\}$ and the borders of $\overline{U_j(\xx)}$, which further form a partition of $\overline{U_j(\xx)}$.

In fact, \eqref{eq:arbitrary-BR-condition} is equivalent to \eqref{eq:polytopal-BR-condition},  meaning that the extra condition imposed on $\fBR_{\text{P}}$ does not compromise its power: if $(\xx, j)$ can be induced by an arbitrary $\fBR$ then it can also be induced by $\fBR_{\text{P}}$.
We state this result in Lemma~\ref{lmm:M-V}.

\begin{lemma}
\label{lmm:M-V}
$u^L(\xx, j) \ge M$ if and only if $u^L(\xx, j) \ge \max_{\yy\in \overline{U_j(\xx)}} \min_{\ell \in [n]\setminus\{j\}} u^L(\yy, \ell)$.
\end{lemma}

\begin{proof}
Recall that we want to show that $u^L(\xx,j) \geq M$ if and only if
\begin{equation}
\label{eq:V-M}
u^L(\xx,j) \ge \max_{\yy \in \overline{U_j(\xx)}} \min_{\ell \in [n] \setminus \{j\}} u^L(\yy,\ell)
\end{equation}
where $M$ is the maximin utility of the leader.

We show that \eqref{eq:V-M} does not hold if and only if $u^L(\xx, j) < M$. Suppose that \eqref{eq:V-M} does not hold. Then $u^L(\xx,j) < \max_{\yy \in \overline{U_j(\xx)}} \min_{\ell \in [n] \setminus \{j\}} u^L(\yy,\ell)$ by definition, which implies that $U_j(\xx) \neq \emptyset$. By the continuity of $\min_{\ell \in [n] \setminus \{j\}} u^L(\cdot,\ell)$, there exists $\yy^* \in U_j(\xx)$ such that
\begin{align*}
u^L(\xx,j) < \min_{\ell \in [n] \setminus \{j\}} u^L(\yy^*,\ell).
\end{align*}
By the definition of $U_j(\xx)$, we also have $u^L(\xx,j) < u^L(\yy^*,j)$. Thus,
\begin{align*}
u^L(\xx,j) < \min_{\ell \in [n]} u^L(\yy^*,\ell) \le \max_{\yy \in \Delta^{m-1}} \min_{\ell \in [n]} u^L(\yy,\ell) = M.
\end{align*}

Conversely, suppose that $u^L(\xx,j) < M$. Let
$\yy^* \in {\arg\max}_{\yy \in \Delta^{m-1}} \min_{\ell \in [n]} u^L(\yy,\ell)$.
Thus, $M = \min_{\ell \in [n]} u^L(\yy^*,\ell)$, and we have
\begin{align*}
u^L(\xx,j) < M = \min_{\ell \in [n]} u^L(\yy^*,\ell) \leq u^L(\yy^*,j)
\end{align*}
which implies that $\yy^* \in U_j(\xx)$. It follows that $M = \max_{\yy \in \overline{U_j(\xx)}} \min_{\ell \in [n]} u^L(\yy,\ell)$ and thus
\begin{align*}
u^L(\xx,j) < \max_{\yy \in \overline{U_j(\xx)}} \min_{\ell \in [n]} u^L(\yy,\ell) \le \max_{\yy \in \overline{U_j(\xx)}} \min_{\ell \in [n]\setminus \{j\}} u^L(\yy,\ell),
\end{align*}
so \eqref{eq:V-M} does not hold.
\end{proof}

Using Lemma~\ref{lmm:M-V}, we can efficiently compute the best strategy profile that can be induced by $\fBR_{\text{P}}$, simply by solving the following Linear Program (LP) for each $j \in [n]$:
\begin{align}
\label{eq:LP-M}
&\max_{\xx \in \Delta^{m-1}}  \quad u^F(\xx, j)  \\
&\text{subject to} \quad 
u^L(\xx,j) \geq M \nonumber
\end{align}

\begin{figure}[t]
\center
\begin{tikzpicture}
\node at (-5,1.1) {$u^L = 
\begin{pmatrix*}[r]
0 & 1    & 1 \\
1 & -1/2 & 1/2 \\
1 & 1/2  & -1/2
\end{pmatrix*}$};

\begin{axis}[
	width= 50mm,
	xtick={0,1},
	ytick={0,1},
	hide obscured y ticks = false,
	ztick=\empty,
	xticklabels={$0$,$1$},
	yticklabels={$0$,$1$},
	xtick pos = right,
	area style,
	view={90}{-90},
	grid=none,
	clip=false,
	axis lines=center,
	xlabel={$y_1$}, ylabel={$y_2$}, 
	xmin = 0, xmax = 1.35,
	ymin = 0, ymax = 1.3,
]

\addplot3 coordinates {
	(1,0,0)
	(0.5, 0.5, 0)
	(0.5, 0, 0.5)
};

\addplot3 coordinates {
	(0.5, 0.25, 0.25)
	(0, 0.5, 0.5)
	(0, 0, 1)
	(0.5, 0, 0.5)
};

\addplot3 coordinates {
	(0.5, 0.25, 0.25)
	(0.5, 0.5, 0)
	(0, 1, 0)
	(0, 0.5, 0.5)
};

\node at (axis cs:0.7,0.15,0.15) {$R_1$};
\node at (axis cs:0.2,0.55,0.25) {$R_2$};
\node at (axis cs:0.2,0.2, 0.6 ) {$R_3$};


\node[circle,inner sep=1.2pt,fill=black] at (axis cs:0.5, 0.5, 0) {};
\node at (axis cs:0.55, 0.6, -0.15) {$\xx$};

\node[circle,inner sep=1.2pt,fill=black] at (axis cs:0.5, 0.25, 0.25) {};
\node at (axis cs:0.42, 0.22, 0.36) {$\zz$};

\node[circle,inner sep=1.2pt,fill=black] at (axis cs:0.5, 0, 0.5) {};
\node at (axis cs:0.55, -0.1, 0.55) {$\ww$};
\end{axis}
\end{tikzpicture}
\caption{No payoff matrix $\uff$ realizes the polytopal BR correspondence $\fBR_{\text{P}}$, such that $\ell \in \fBR_\text{P}$ if and only if $\yy \in R_\ell$, where 
$R_1 = \{\yy \in \Delta^2: y_1 \ge y_2 + y_3\}$,
$R_2 = \{\yy \in \Delta^2: y_1 \le y_2 + y_3 \text{ and } y_2 \ge y_3 \}$, and
$R_3 = \{\yy \in \Delta^2: y_1 \le y_2 + y_3 \text{ and } y_2 \le y_3 \}$. \label{fig:no-payoff-matix}}
\end{figure}
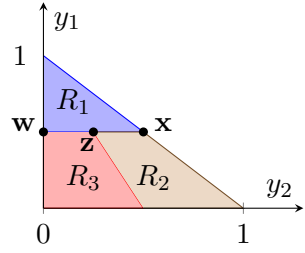

At this point, it might be tempting to think that with the polytopal constraint imposed, we would also be able to construct an explicit payoff matrix $\uff$ to implement $\fBR_{\text{P}}$. 
Unfortunately, this is not true as Example~\ref{exm:no-payoff-matrix} illustrates. 
Surprisingly though, in the next section we will show that, even though we cannot construct a payoff matrix that implements $\fBR_{\text{P}}$ directly, every strategy profile $(\xx,j)$ that is $\fBR_{\text{P}}$-inducible, is in fact payoff-inducible. We also present an efficient algorithm for computing a payoff matrix $\uff$ to induce such $(\xx, j)$.

\begin{example}
\label{exm:no-payoff-matrix}
Consider a $3\times 3$ game with the leader payoff matrix given in Figure~\ref{fig:no-payoff-matix}.
Let $\fBR_{\text{P}}$ be a polytopal BR correspondence defined by the regions $R_1$, $R_2$, and $R_3$ in Figure~\ref{fig:no-payoff-matix}, such that $\ell \in \fBR_\text{P}$ if and only if $\yy \in R_\ell$.
This best response behavior cannot be realized by any payoff matrix. 
To see this, suppose $\fBR_{\text{P}}$ is realized by some $\uff \in \mathbb{R}^{3\times 3}$.
Let $\xx = (\frac{1}{2}, \frac{1}{2}, 0)$, $\ww = (\frac{1}{2}, 0, \frac{1}{2})$, and $\zz = (\frac{1}{2}, \frac{1}{4}, \frac{1}{4})$.
We have $\fBR_{\text{P}}(\zz) = \{1,2,3\}$ and $\fBR_{\text{P}}(\ww) = \{1,3\}$.
This means 
that $u^L(\zz,\, 1) = u^L(\zz,\, 3) = u^L(\zz,\, 2)$ and
$u^L(\ww, 1) = u^L(\ww, 3) > u^L(\ww, 2)$.
Since $\xx = 2 \zz - \ww$, by the linearity of the utility function,
$u^L(\xx, 1) = u^L(\xx, 3) < u^L(\xx, 2)$,
which contradicts the fact that $\fBR_{\text{P}}(\xx) = \{1,2\}$.
\hfill $\qed$
\end{example}

\section{Payoff Inducibility} 
\label{sec:payoff}

In this section, we will show that every profile strategy that can be induced by $\fBR_{\text{P}}$ is also payoff-inducible, and a corresponding payoff matrix can be efficiently constructed. Recall that the maximin utility of the leader is denote by $M = \max_{\yy\in\Delta^{m-1}} \min_{\ell \in [n]} u^L(\yy, \ell)$. 
We will show the following characterization as one of our key results, which enables us to use the LP in \eqref{eq:LP-M} to efficiently compute a payoff matrix that achieves the optimal inducible utility.

\begin{theorem}
\label{thm:payoff-inducibility-M}
A strategy profile $(\xx,j)$ is payoff-inducible if and only if $u^L(\xx,j) \geq M$. Furthermore, a matrix $\uff$ inducing $(\xx, j)$ can be constructed in polynomial time. 
\end{theorem}

One direction of the characterization is easy to show. Indeed, if $(\xx,j)$ is payoff-inducible, then it is also BR-inducible, and as seen in Section~\ref{sec:best-response}, it holds that $u^L(\xx,j) \geq M$.

Now consider any profile $(\xx,j)$ such that $u^L(\xx,j) \geq M$.
Recall that $U_j(\xx) = \{\yy \in \Delta^{m-1} : u^L(\yy,j) > u^L(\xx,j)\}$.
Without loss of generality, in what follows, we can also assume that $U_j(\xx) \neq \emptyset$: if $U_j(\xx) = \emptyset$, then $(\xx,j)$ will be an SSE if the follower always responds by playing $j$; this can easily be achieved by claiming that $j$ strictly dominates all other strategies, i.e., by letting $\uff(i,j) = 1$ and $\uff(i,\ell)=0$ for all $i \in [m]$ and $\ell \in [n] \setminus \{j\}$.

We begin by analyzing the following payoff function that forms the basis of our approach.
Let $\widehat{S} \subseteq [n] \setminus \{j\}$ and pick $k \in \argmin_{\ell \in \widehat{S}} u^L(\xx,\ell)$ arbitrarily. For all $\yy \in \Delta^{m-1}$, let
\begin{equation}
\label{eq:uF-theta-def}
\uff(\yy,\ell) = \left\{\begin{tabular}{ll}
$-u^L(\yy,\ell)$ & if $\ell \in \widehat{S}$\\
$-u^L(\yy,k) - 1$ & if $\ell \in [n] \setminus (\widehat{S} \cup \{j\})$\\
$-u^L(\yy,k) + \alpha \left(u^L(\xx,j) - u^L(\yy,j)\right)$ & if $\ell=j$
\end{tabular}\right.
\end{equation}
where $\alpha > 0$ is a constant. 
In what follows, we will let $\fBR$ denote the best response correspondence corresponding to $\uff$, i.e., $\fBR(\yy) = \argmax_{\ell\in[n]} \uff(\yy, \ell)$.
Note that we can compute the payoff matrix corresponding to $\uff$ in polynomial time.
Then, the hope is that with appropriately chosen $\widehat{S}$ and $\alpha$, the payoff matrix will induce $(\xx, j)$.  
Indeed, $\uff$ has the following nice properties:

\begin{enumerate}
\item[i.]
Strategy $j$ is indeed a best response to $\xx$, since, by the choice of $k$ we have
$$\uff(\xx,j)
= - u^L(\xx,k)
\geq - \min_{\ell \in \widehat{S}} u^L(\xx,\ell)
= \max_{\ell \in \widehat{S}} \uff(\xx,\ell).
$$

\item[ii.] 
Any $\ell \in [n] \setminus (\widehat{S} \cup \{j\})$ cannot be a best response of the follower as it is strictly dominated by $k$, i.e., $\uff(\yy, \ell) < \uff(\yy, k)$ for all $\yy \in \Delta^{m-1}$. Thus, $\fBR(\yy) \subseteq \widehat{S} \cup \{j\}$ for all $\yy \in \Delta^{m-1}$.
	
\item[iii.]
If $j$ is a best response to some $\yy \in \Delta^{m-1}$, then $u^L(\yy, j) \le u^L(\xx, j)$. Indeed, $j \in \fBR(\yy)$ implies that 
$$\uff(\yy,j) = \max_{\ell \in [n]} \uff(\yy,\ell) \geq \uff(\yy,k).$$
Substituting $\uff(\yy, j) = -u^L(\yy,k) + \alpha \left(u^L(\xx,j) - u^L(\yy,j)\right)$ into this inequality and rearranging the terms immediately gives $u^L(\yy, j) \le u^L(\xx, j)$.

\item[iv.]
If any $\ell \in \widehat{S}$ is a best response to some $\yy \in \Delta^{m-1}$, then it holds that
$\uff(\yy,\ell) = \max_{\ell' \in  \widehat{S}} \uff(\yy,\ell')$, which implies that
\begin{align}
\label{eq:uL-eq-minoverhatS}
u^L(\yy,\ell) = \min_{\ell' \in \widehat{S}} u^L(\yy,\ell').
\end{align}
\end{enumerate}

Therefore, if the following also holds for the $\yy$ in (iv), 
\begin{align*}
\min_{\ell' \in \widehat{S}} u^L(\yy,\ell') \le u^L(\xx, j),
\end{align*}
then by \eqref{eq:uL-eq-minoverhatS} we will have $u^L(\yy, \ell) \le u^L(\xx, j)$ for every $\ell \in \fBR(\yy) \cap \widehat{S}$. This, together with (ii) and (iii), will imply that $u^L(\xx, j) \ge u^L(\yy, \ell)$ for every $\ell \in \fBR(\yy)$. Therefore, $(\xx, j)$ will indeed form an SSE given that $j \in \fBR(\xx)$ by (i).
We state this observation as the following lemma.

\begin{lemma}
\label{lmm:key-condition}
If $\min_{\ell' \in \widehat{S}} u^L(\yy,\ell') \le u^L(\xx, j)$ holds for all $\yy\in\Delta^{m-1}$ such that $\fBR(\yy) \cap \widehat{S} \neq \emptyset$, then the payoff matrix defined by \eqref{eq:uF-theta-def} induces $(\xx, j)$.
\end{lemma}

The proof of Theorem~\ref{thm:payoff-inducibility-M} is then completed by showing the following result.

\begin{proposition}\label{prp:key-condition-holds}
If $u^L(\xx,j) \geq M$ and $U_j(x) \neq \emptyset$, then we can construct $\widehat{S} \subseteq [n] \setminus \{j\}$ and $\alpha > 0$ in polynomial time, with which the condition of Lemma~\ref{lmm:key-condition} holds for $\uff$ as defined in \eqref{eq:uF-theta-def}.
\end{proposition}

The proof relies on the following useful lemma.

\begin{lemma}[Farkas' Lemma~\citep{Boyd2014convex}]\label{lmm:farkas}
Let $\mathbf{A} \in \mathbb{R}^{n_1 \times n_2}$ and $\bb \in \mathbb{R}^{n_1}$. Then exactly one of the following statements is true:
\begin{enumerate}
	\item there exists $\zz \in \mathbb{R}^{n_2}$ such that $\mathbf{A} \zz= \bb$ and $\zz \geq 0$;
	\item there exists $\zz \in \mathbb{R}^{n_1}$ such that $\mathbf{A}^\mathsf{T} \zz \geq 0$ and $\bb \cdot \zz < 0$.
\end{enumerate}
\end{lemma}

\begin{proof}[Proof of Proposition~\ref{prp:key-condition-holds}]
Consider any strategy profile $(\xx,j)$ with $u^L(\xx,j) \geq M$ and $U_j(x) \neq \emptyset$.
We begin by taking care of a simple case, as an immediate corollary of Lemma~\ref{lmm:key-condition}.
\begin{corollary}
\label{crl:easy-case}
A matrix $\uff$ that induces $(\xx,j)$ can be constructed in polynomial time if it holds that
\begin{align}
\label{eq:easy-case-condition}
u^L(\xx,j) \geq M_{-j} := \max_{\yy \in \Delta^{m-1}} \min_{\ell \in  [n] \setminus \{j\}} u^L(\yy,\ell).
\end{align}
\end{corollary}
\begin{proof}
Let $\widehat{S} = [n] \setminus \{j\}$. Then, for every $\yy \in \Delta^{m-1}$, we immediately obtain that
\begin{align*}
u^L(\xx,j) \geq \max_{\yy \in \Delta^{m-1}} \min_{\ell \in  [n] \setminus \{j\}} u^L(\yy,\ell) \ge \min_{\ell \in  \widehat{S}} u^L(\yy,\ell)
\end{align*} 
By Lemma~\ref{lmm:key-condition}, the payoff matrix defined by \eqref{eq:uF-theta-def} (with, say, $\alpha=1$) then induces $(\xx, j)$, and can clearly be computed in polynomial time.
\end{proof}

The more challenging case is when \eqref{eq:easy-case-condition} does not hold (e.g., the case with the profile $(\xx, 1)$ in Example~\ref{exm:no-payoff-matrix}).
In what follows, we prove Proposition~\ref{prp:key-condition-holds} by showing that there is still a choice of $\widehat{S}$ and $\alpha$ that leads to the condition in Lemma~\ref{lmm:key-condition}, even when \eqref{eq:easy-case-condition} does not hold. Thus, from now on, we assume that
\begin{equation}\label{eq:assumption-M-j}
u^L(\xx,j) < M_{-j}.
\end{equation}

We define the following useful components. By Lemma~\ref{lmm:M-V} and the assumption that $u^L(\xx, j) \ge M$, we know that
\begin{equation}\label{eq:at-least-V}
u^L(\xx,j) \geq V
\end{equation}
where
\begin{align*}
V = \max_{\yy \in \overline{U_j(\xx)}} \min_{\ell \in [n] \setminus \{j\}} u^L(\yy,\ell).
\end{align*}
Since $\overline{U_j(\xx)}\neq \emptyset$, there exists $\yy^* \in \overline{U_j(\xx)}$ such that
\begin{equation}
\label{eq:def-y-star}
\min_{\ell \in [n] \setminus \{j\}} u^L(\yy^*,\ell) = V,
\end{equation}
which can be computed efficiently by solving an LP (i.e., maximize $\mu$, subject to $\mu \le u^L(\yy, \ell)$ for all $\ell \in [n] \setminus \{j\}$ and $\yy \in \overline{U_j(\xx)}$).
We then let 
\begin{align*}
S = \{\ell \in [n] \setminus \{j\} \, | \, u^L(\yy^*,\ell) = V\}.
\end{align*}

Before we proceed, we prove two useful technical results.
\begin{lemma}
\label{lmm:key-claim-yystar}
$u^L(\yy^*,j) = u^L(\xx,j)$.
\end{lemma}
\begin{proof}[Proof]

For the sake of contradiction, suppose that $u^L(\yy^*,j) \neq u^L(\xx,j)$. Since $\yy^* \in \overline{U_j(\xx)}$, we have that $u^L(\yy^*,j) \ge u^L(\xx,j)$, so it must be that $u^L(\yy^*,j) > u^L(\xx,j)$.

The assumption \eqref{eq:assumption-M-j} that $u^L(\xx,j) < M_{-j}$ implies that there exists $\hat{\yy} \in \Delta^{m-1}$ such that 
\[
\min_{\ell \in [n]\setminus\{j\}} u^L(\hat{\yy}, \ell) > u^L(\xx,j) \ge V,
\]
where we also use \eqref{eq:at-least-V}. 
Now that $\min_{\ell \in [n]\setminus\{j\}} u^L(\yy^*, \ell) = V$ by \eqref{eq:def-y-star}, by the concavity of $\min_{\ell\in[n]\setminus\{j\}} u^L(\cdot, \ell)$, it follows that $\min_{\ell\in[n]\setminus\{j\}} u^L(\zz, \ell) > V$ for all $\zz$ on the segment $[\hat{\yy}, \yy^*)$; $\zz\in\Delta^{m-1}$ as $\Delta^{m-1}$ is convex.
Now that we have $u^L(\yy^*,j) > u^L(\xx,j)$ under our assumption, when $\zz$ is sufficiently close to $\yy^*$, we can have $u^L(\zz,j) \ge u^L(\xx,j)$ and hence, $\zz \in \overline{U_j(\xx)}$. This leads to the contradiction that 
\begin{align*}
V 
=\max_{\yy\in \overline{U_j(\xx)}} \min_{\ell\in[n]\setminus\{j\}} u^L(\yy, \ell) 
\ge \min_{\ell\in[n]\setminus\{j\}} u^L(\zz, \ell) 
> V. &\qedhere
\end{align*}
\end{proof}

\begin{lemma}
\label{lmm:super-claim}
$\min_{\ell \in S} u^L(\yy,\ell) < V$ for all $\yy \in U_j(\xx)$.
\end{lemma}

\begin{proof}

For the sake of contradiction, assume that there exists $\hat{\yy} \in U_j(\xx)$ such that 
$$\min_{\ell \in S} u^L(\hat{\yy},\ell) \geq V.$$ 
By assumption \eqref{eq:assumption-M-j} that $u^L(\xx,j) < M_{-j}$, there exists $\zz \in \Delta^{m-1}$ such that
$\min_{\ell \in [n] \setminus \{j\}} u^L(\zz,\ell) > u^L(\xx,j) \geq V$, which immediately yields the following given that $S \subseteq [n]\setminus\{j\}$ by definition:
\[
\min_{\ell \in S} u^L(\zz,\ell) > V.
\]
By definition, $u^L(\yy^*, \ell) = V$ for all $\ell \in S$, which also implies that $u^L(\yy^*, \ell) > V$ for all $\ell \in [n] \setminus (\{j\} \cup S)$ (otherwise, we would have $\min_{\ell\in[n]\setminus\{j\} u^L(\yy^*, \ell)} < V$).
Thus, we have 
\[
\min_{\ell \in S} u^L(\yy^*,\ell) = V
\quad \text{ and } \quad 
\min_{\ell \in [n] \setminus (\{j\} \cup S)} u^L(\yy^*,\ell) > V.
\]

Now consider a point $\ww$ on the segment $(\yy^*, \hat{\yy}]$. 
Since $\yy^* \in \overline{U_j(\xx)}$ and $\hat{\yy} \in U_j(\xx)$, i.e., $u^L(\yy^*, j) \ge u^L(\xx, j)$ and $u^L(\hat{\yy}, j) > u^L(\xx, j)$, 
we have $u^L(\ww, j) > u^L(\xx, j)$ and hence, $\ww \in U_j(\xx)$.
In addition, by continuity, when $\ww$ is sufficiently close to $\yy^*$, we have 
\begin{align}
\label{eq:ww-1}
\min_{\ell \in [n] \setminus (\{j\} \cup S)} u^L(\ww,\ell) > V.
\end{align}
By concavity of the function $\min_{\ell \in S} u^L(\cdot,\ell)$, since $\min_{\ell \in S} u^L(\yy,\ell) \ge V$ for both $\yy \in \{\yy^*, \hat{\yy}\}$, we have 
\begin{align}
\label{eq:ww-2}
\min_{\ell \in S} u^L(\ww,\ell) \ge V.
\end{align}

Analogously, we can find a point $\ww'\in U_j(\xx)$ on the segment $(\ww, \zz]$, such that \eqref{eq:ww-1}
and \eqref{eq:ww-2} hold for $\ww'$ while \eqref{eq:ww-2} is strict, in particular. Thus, we have
\begin{align*}
\min_{\ell \in [n]\setminus\{j\}} u^L(\ww',\ell) 
> V = \max_{\yy \in \overline{U_j(\xx)}} \min_{\ell \in [n]\setminus\{j\}} u^L(\yy,\ell),
\end{align*}
which is a contradiction as $\ww'\in U_j(\xx)$.
\end{proof}

In what follows, we use the coordinates $(y_1,\dots,y_{m-1})$ for every point $\yy \in \Delta^{m-1}$, i.e., we have
\begin{equation*}
\Delta^{m-1} =
\left\{(y_1,\dots,y_{m-1}) \in \mathbb{R}_{\geq 0} \, :\, \sum_{i=1}^{m-1} y_i \leq 1\right\}.
\end{equation*}
Accordingly, we can write the utility function as 
\begin{align*}
u^L(\yy,\ell)
= \gv_\ell \cdot \yy + u^L(m, \ell),
\end{align*}
where $\gv_\ell \in \mathbb{R}^{m-1}$ and its $i$-th component is $g_{\ell, i} = u^L(i, \ell) - u^L(m, \ell)$;  ``$\cdot$'' denotes the inner product.
Hence, we have
\begin{align}
\label{eq:rewrite-ul}
u^L(\yy,\ell) = \gv_\ell \cdot (\yy-\yy^*) + u^L(\yy^*,\ell) =
\begin{cases}
\gv_\ell \cdot (\yy-\yy^*) + V %
& \text{if } \ell \in S \\
\gv_j \cdot (\yy-\yy^*) + u^L(\xx,j) %
& \text{if } \ell = j
\end{cases}
\end{align}
where $u^L(\yy^*,\ell) = V$ for all $\ell \in S$ by the definition of $S$, and $u^L(\yy^*,j) = u^L(\xx,j)$ by Lemma~\ref{lmm:key-claim-yystar}.
Note that since $U_j(\xx) \neq \emptyset$, it must be that $\gv_j \neq 0$.

We also write the $m$ boundary conditions that define $\Delta^{m-1}$ as $\ee_i \cdot \yy \ge \beta_i$. Namely, for each $i \in [m-1]$, let $\ee_i \in \mathbb{R}^{m-1}$ be the $i$-th unit vector and $\beta_i=0$, while  $\ee_{m} = (-1, \dots, -1) \in \mathbb{R}^{m-1}$ and $\beta_m = -1$.
Thus, $\Delta^{m-1} = \{\yy \in \mathbb{R}^{m-1} \, : \, \ee_i \cdot \yy \geq \beta_i \text{ for } i \in [m] \}$.
Let 
\begin{align*}
B =  \{i \in [m] \, : \, \ee_i \cdot \yy^* = \beta_i\}
\end{align*}
be the set of boundary conditions that are tight for $\yy^*$. 
Note that for any $\yy \in \Delta^{m-1}$ we have
\begin{equation}
\label{eq:boundary-property}
\ee_i \cdot (\yy-\yy^*) \geq 0 \quad \text{ for all } i \in B.
\end{equation}

We can now prove the following result using Farkas' Lemma (Lemma~\ref{lmm:farkas}), which allows us to express $-\gv_j$ as a non-negative linear combination of $\gv_\ell$'s and $\ee_i$'s.

\begin{lemma}
\label{lmm:non-neg-linear-comb}
$-\gv_j$ can be expressed as a non-negative linear combination of $\{\gv_\ell  :  \ell \in S\} \cup \{\ee_i  :  i \in B\}$, i.e.
$-\gv_j = \sum_{\ell \in S} \lambda_\ell \gv_\ell + \sum_{i \in B} \mu_i \ee_i$,
where $\lambda_\ell \geq 0$ and $\mu_i \geq 0$. 
\end{lemma}

\begin{proof}
We use Farkas' Lemma (Lemma~\ref{lmm:farkas}) and let $n_1 = m-1$ and $n_2 = |S| + |B|$. The columns of $\mathbf{A}$ are exactly the vectors $\{\gv_\ell  :  \ell \in S\} \cup \{\ee_i  :  i \in B\}$. We set $\bb = -\gv_j$. Note that the first alternative of Farkas' Lemma immediately yields the statement we want to prove. Thus, we set out to prove that the second alternative cannot hold.

Assume, for the sake of contradiction, that there exists $\zz \in \mathbb{R}^{m-1}$ such that $\mathbf{A}^\mathsf{T} \zz \geq 0$ and $\bb \cdot \zz < 0$, i.e., $\gv_\ell \cdot \zz \geq 0$ for all $\ell \in S$, $\ee_i \cdot \zz \geq 0$ for all $i \in B$, and $\gv_j \cdot \zz > 0$.

Then, by picking $\delta > 0$ sufficiently small, it holds for $\yy = \yy^* + \delta \zz$ that:
\begin{itemize}
\item 
By~\eqref{eq:rewrite-ul}, we have the following for all $\ell \in S$:
$$u^L(\yy,\ell) = \gv_\ell \cdot (\yy - \yy^*) + V = \delta \gv_\ell \cdot \zz + V \geq V.$$
In addition,
$$u^L(\yy,j) = \gv_j \cdot (\yy-\yy^*) + u^L(\yy^*,j) = \delta \gv_j \cdot \zz + u^L(\xx,j) > u^L(\xx,j).$$

\item 
$\yy \in \Delta^{m-1}$: For $i \in B$, we immediately obtain that $\ee_i \cdot \yy = \ee_i \cdot (\yy^*+\delta \zz) \geq \ee_i \cdot \yy^* = \beta_i$, which means that these boundary conditions are satisfied. For $i \in [m] \setminus B$, we know that $\ee_i \cdot \yy^* > \beta_i$ and thus by picking $\delta > 0$ small enough, we can ensure that $\ee_i \cdot \yy = \ee_i \cdot \yy^* + \delta (\ee_i \cdot \zz) \geq \beta_i$. 

\end{itemize}
Thus, it follows that $\yy \in U_j(\xx)$ and $\min_{\ell \in S} u^L(\yy,\ell) \geq V$. But this cannot hold according to Lemma~\ref{lmm:super-claim}.
\end{proof}

We can now complete the proof of Proposition~\ref{prp:key-condition-holds}.
We first express $-\gv_j$ as a non-negative linear combination of the vectors $\{\gv_\ell : \ell \in S\} \cup \{\ee_i : i \in B\}$. By Lemma~\ref{lmm:non-neg-linear-comb} we know that this is possible and it is easy to see that we can find the coefficients in polynomial time (e.g. by solving an LP).
We thus obtain $-\gv_j = \sum_{\ell \in S} \lambda_\ell \gv_\ell + \sum_{i \in B} \mu_i \ee_i$, where $\lambda_\ell \ge 0$ for every $\ell \in S$ and $\mu_i \ge 0$ for every $i \in B$.
Let $\widehat{S} = \{\ell \in S : \lambda_\ell > 0\}$. 
We will argue that $\widehat{S} \neq \emptyset$.

Observe that since now $-\gv_j = \sum_{\ell \in S} \lambda_\ell \gv_\ell + \sum_{i \in B} \mu_i \ee_i$ and, by \eqref{eq:boundary-property}, we have $\ee_i \cdot (\yy-\yy^*) \geq 0$ for all $\yy\in\Delta^{m-1}$ and $i \in B$, it follows that, for all $\yy\in\Delta^{m-1}$, we have
\begin{align}
-\gv_j \cdot (\yy-\yy^*) 
&= \sum_{\ell \in S} \lambda_\ell \gv_\ell \cdot (\yy-\yy^*) + \sum_{i \in B} \mu_i \ee_i \cdot (\yy-\yy^*) \nonumber\\
&\ge \sum_{\ell \in S} \lambda_\ell \gv_\ell \cdot (\yy-\yy^*) \nonumber\\
&= \sum_{\ell \in \widehat{S}} \lambda_\ell \gv_\ell \cdot (\yy-\yy^*),
\label{eq:lambda-g-hatS}
\end{align}
where the last transition is due to the fact that $\lambda_\ell = 0$ for all $\ell \in S\setminus \widehat{S}$, as implied by the definition of $\widehat{S}$.

Since $U_j(\xx) \neq \emptyset$, consider any $\yy \in U_j(\xx)$. By definition, this means that $u^L(\yy,j) > u^L(\xx,j)$, which further implies that $\gv_j \cdot (\yy-\yy^*) > 0$ since $u^L(\yy,j) = \gv_j \cdot (\yy-\yy^*) + u^L(\xx,j)$ by \eqref{eq:rewrite-ul}.
By \eqref{eq:lambda-g-hatS}, we then have 
\[
\sum_{\ell \in \widehat{S}} \lambda_\ell \gv_\ell \cdot (\yy-\yy^*) < 0.
\]
Hence, $\widehat{S} \neq \emptyset$.

It remains to show that with the above $\widehat{S}$ and, in particular, $\alpha = 1/\lambda_k$ (recall that $k \in \argmin_{\ell \in \widehat{S}} u^L(\xx,\ell)$), the condition in Lemma~\ref{lmm:key-condition} holds, i.e., we prove that $\min_{\ell \in \widehat{S}} u^L(\yy,\ell) \le u^L(\xx, j)$ holds for all $\yy\in\Delta^{m-1}$ such that $\fBR(\yy) \cap \widehat{S} \neq \emptyset$.

For the sake of contradiction, suppose that there exists $\yy \in \Delta^{m-1}$ such that $\fBR(\yy) \cap \widehat{S} \neq \emptyset$, but $u^L(\yy, \ell) > u^L(\xx, j)$ for all $\ell \in \widehat{S}$. 
By \eqref{eq:at-least-V}, we have $u^L(\xx, j) \ge V$, and
thus $u^L(\yy,\ell) > V$ for all $\ell \in \widehat{S}$.
By \eqref{eq:rewrite-ul}, we have $u^L(\yy,\ell) = \gv_\ell \cdot (\yy-\yy^*) + V$; thus, $\gv_\ell \cdot (\yy-\yy^*) > 0$ for all $\ell \in \widehat{S}$.

Using \eqref{eq:lambda-g-hatS} and the fact that $k \in \widehat{S}$ by our choice, we then obtain
\begin{align*}
-\gv_j \cdot (\yy-\yy^*) 
\geq \sum_{\ell \in \widehat{S}} \lambda_\ell \gv_\ell \cdot (\yy-\yy^*)
\geq \lambda_k \gv_k \cdot (\yy-\yy^*).
\end{align*}
By \eqref{eq:rewrite-ul}, we have
\begin{align*}
u^L(\xx,j) - u^L(\yy,j)
= -\gv_j \cdot (\yy-\yy^*).
\end{align*}
Recall that it is defined that $\uff(\yy,j) 
= -u^L(\yy,k) + \alpha \left(u^L(\xx,j) - u^L(\yy,j)\right)$ as in \eqref{eq:uF-theta-def}.
Using the above two equations and \eqref{eq:rewrite-ul}, we then obtain the following:
\begin{align*}
\uff(\yy,j) 
&= -u^L(\yy,k) + \alpha \left(u^L(\xx,j) - u^L(\yy,j)\right) \\
&= -\gv_k \cdot (\yy-\yy^*) - V - \alpha \gv_j \cdot (\yy-\yy^*) \\
&\geq -V + (\alpha \lambda_k - 1) \gv_k \cdot (\yy-\yy^*)\\ 
&= -V.
\end{align*}
However, by \eqref{eq:uF-theta-def} we also have $\uff(\yy,\ell) = -u^L(\yy, \ell)$ if $\ell \in \widehat{S}$, which implies that for all $\ell \in \widehat{S}$ it holds that
\[
\uff(\yy,j) \ge - V > -u^L(\yy, \ell) = \uff(\yy,\ell).
\]
Hence, $\fBR(\yy) \cap \widehat{S} = \emptyset$, which contradicts our assumption.
\end{proof}

\section{Robustness with Respect to Equilibrium Selection}
\label{sc:strong}

As discussed in Section~\ref{sec:prelim}, a weakness of BR- and payoff-inducible strategy profiles is that the resulting games may have multiple SSEs, in which case the follower depends on the leader to choose the SSE that maximizes his utility. To avoid this, in this section, we turn our attention to strong inducibility (see Definition~\ref{def:strong-inducible}) and attempt to find a payoff matrix $\uff$ such that $\widetilde{\mathcal{G}}$ has a unique SSE.

We begin with an example showcasing that, in general, the best strongly inducible profile can be much worse than the best payoff-inducible profile.

\begin{example}
\label{exm:degenerate}
Consider a $3 \times 2$ game $\calG=(u^L,u^F)$ with the payoff matrices given in Figure~\ref{fig:degenerate}. 
Note that the follower obtains positive utility only by playing his strategy $1$. 
Now, observe that the SSE $(\xx^*, 1)$, $\xx^*=(0,0,1) \in \Delta^2$, is payoff-inducible and yields a utility of $1$ for the follower: it can be induced by any payoff matrix in which strategy $1$ of the follower strictly dominates all other strategies.
However, such a payoff matrix will also induce other SSEs, e.g., $(\yy^*,1)$ with $\yy^*=(1,0,0) \in \Delta^2$. Indeed, it holds that no profile of the form $(\yy,1)$ can be \emph{strongly} induced, and thus the optimal utility the follower can obtain at a strongly inducible profile is $0$. To see this, first note that, as seen above, if the follower claims that strategy $1$ is his unique best response for all points in $\Delta^2$, then the SSE is not unique. On the other hand, if strategy $2$ is a best response at some point $\zz \in \Delta^2$, then $(\yy,1)$ will not be an SSE, since for the leader $u^L(\yy,1) < u^L(\zz,2)$ for any $\yy, \zz \in \Delta^2$.
\hfill $\qed$
\end{example}

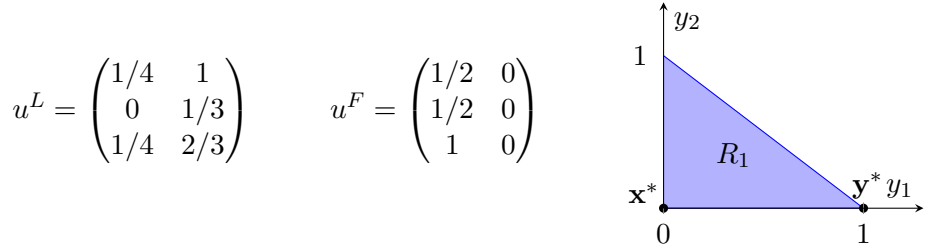
\begin{figure}[t]
\center
\begin{tikzpicture}
\node at (-7,1.3) {$u^L = 
\begin{pmatrix}
1/4 & 1  \\
0 & 1/3 \\
1/4 & 2/3 
\end{pmatrix}$};

\node at (-3,1.3) {$u^F = 
\begin{pmatrix}
1/2 & 0  \\
1/2 & 0 \\
1 & 0 
\end{pmatrix}$};

\begin{axis}[
	width= 50mm,
	xtick={0,1},
	ytick={0,1},
	hide obscured y ticks = false,
	ztick=\empty,
	xticklabels={$0$,$1$},
	yticklabels={$0$,$1$},
	xtick pos = right,
	area style,
	view={90}{-90},
	grid=none,
	clip=false,
	axis lines=center,
	xlabel={$y_2$}, ylabel={$y_1$}, 
	xmin = 0, xmax = 1.35,
	ymin = 0, ymax = 1.3,
]

\addplot3 coordinates {
	(1, 0, 0)
	(0, 1, 0)
	(0, 0, 1)
};


\node at (axis cs:0.35,0.35,0.3) {$R_1$};

\node[circle,inner sep=1.2pt,fill=black] at (axis cs:0, 0, 0.5) {};
\node at (axis cs:0.1, -0.1, 0.55) {$\xx^*$};
\node[circle,inner sep=1.2pt,fill=black] at (axis cs:0, 1, 0.5) {};
\node at (axis cs:0.15, 1.02, 0.55) {$\yy^*$};

\end{axis}
\end{tikzpicture}
\caption{A game where the optimal inducible utility is $1$, but the optimal \emph{strongly} inducible utility is $0$. \label{fig:degenerate}}
\end{figure}

The problem in Example~\ref{exm:degenerate} stems from the following observation: if the follower reports a payoff matrix such that strategy $1$ is the unique best response for all points in the domain, then there are multiple SSEs. This can be thought of as a ``degenerate'' case, since it would occur with probability $0$, if the payoffs of the leader were drawn uniformly at random in $[0,1]$. We formalize this as follows.

\begin{definition}\label{def:degenerate}
A leader payoff matrix $u^L$ is said to be \emph{max-degenerate}, if there exists $j \in [n]$ such that $|\argmax_{i \in [m]} u^L(i,j)| > 1$.
\end{definition}

We next provide an example showing that even when $u^L$ is \emph{not} max-degenerate, we cannot hope to \emph{exactly} achieve the optimal inducible utility via a strongly inducible profile.

\begin{example}
\label{exm:only_approx_strong_possible}
Consider a $3 \times 2$ game with the leader and follower payoff matrices given in Figure~\ref{fig:no_opt_ind_payoff}. It is easy to check that $u^L$ is not max-degenerate. Now, observe that the maximin utility of the leader is $M=1/2$ and is achieved at the point $\yy^* = (\frac{1}{2},\frac{1}{2},0) \in \Delta^2$. Let $\xx^* = (0,0,1) \in \Delta^2$. Since $u^L(\xx^*,1) = 1/2 \geq M$, it follows that $(\xx^*,1)$ is payoff-inducible by Theorem~\ref{thm:payoff-inducibility-M}. Indeed, the partition $(R_1,R_2)$ of $\Delta^2$ in Figure~\ref{fig:no_opt_ind_payoff} shows how $(\xx^*,1)$ can be induced. Note that $u^F(\xx^*,1) = 1$, while any profile different from $(\xx^*,1)$ yields utility strictly less than $1$ for the follower. We will now show that $(\xx^*,1)$ cannot be strongly induced, which implies that any strongly inducible profile gives utility strictly less than $1$ to the follower. 
Indeed, suppose that $(\xx^*,1)$ is induced by some $\uff$.
If by $\uff$ strategy $1$ is a best response to $\yy^*$, then $(\xx^*,1)$ cannot be the unique SSE, since $u^L(\xx^*,1) = u^L(\yy^*,1)$. On the other hand, if strategy $2$ is the only best response to $\yy^*$, then there exists some sufficiently small $\delta > 0$ such that strategy $2$ is also a best response to $\ww^*=(\frac{1}{2}-\delta,\frac{1}{2}+\delta, 0)$ (see Figure~\ref{fig:no_opt_ind_payoff}). However, this means that $(\xx^*,1)$ cannot be an SSE, since $u^L(\xx^*,1) = 1/2$ and $u^L(\ww^*,2) = 1/2+\delta$.
\hfill $\qed$
\end{example}

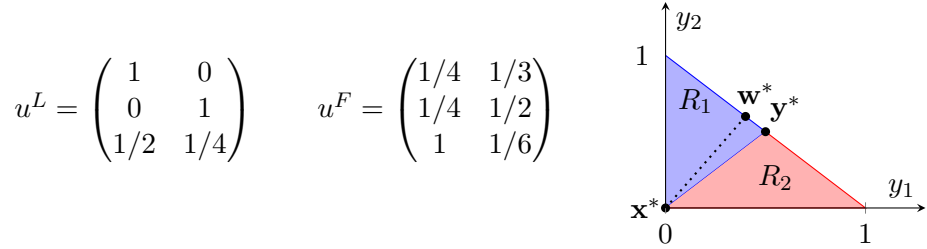
\begin{figure}[t]
\center
\begin{tikzpicture}
\node at (-7,1.3) {$u^L = 
\begin{pmatrix}
1 & 0 \\
0 & 1 \\
1/2 & 1/4
\end{pmatrix}$};

\node at (-3,1.3) {$u^F = 
\begin{pmatrix}
1/4 & 1/3 \\
1/4 & 1/2 \\
1 & 1/6  
\end{pmatrix}$};

\begin{axis}[
	width= 50mm,
	xtick={0,1},
	ytick={0,1},
	hide obscured y ticks = false,
	ztick=\empty,
	xticklabels={$0$,$1$},
	yticklabels={$0$,$1$},
	xtick pos = right,
	area style,
	view={90}{-90},
	grid=none,
	clip=false,
	axis lines=center,
	xlabel={$y_2$}, ylabel={$y_1$}, 
	xmin = 0, xmax = 1.35,
	ymin = 0, ymax = 1.3,
]

\addplot3 coordinates {
	(1,0,0)
	(0.5, 0.5, 0)
	(0, 0, 0)
};

\addplot3 coordinates {
	(0.5, 0.5, 0.25)
	(0, 1, 0.5)
	(0, 0, 1)
	(0, 0, 0.5)
};

\draw[black,dotted,thick] (0, 0, 1) -- (0.6, 0.4, 0.5);


\node at (axis cs:0.7,0.15,0.15) {$R_1$};
\node at (axis cs:0.2,0.55,0.25) {$R_2$};

\node[circle,inner sep=1.2pt,fill=black] at (axis cs:0, 0, 0.5) {};
\node at (axis cs:0, -0.1, 0.55) {$\xx^*$};

\node[circle,inner sep=1.2pt,fill=black] at (axis cs:0.5, 0.5, 0) {};
\node at (axis cs:0.65, 0.6, -0.15) {$\yy^*$};


\node[circle,inner sep=1.2pt,fill=black] at (axis cs:0.6, 0.4, 0) {};
\node at (axis cs:0.75, 0.45, -0.15) {$\ww^*$};

\end{axis}
\end{tikzpicture}
\caption{A non-max-degenerate game for which the optimal inducible utility cannot be achieved by any strongly inducible profile.\label{fig:no_opt_ind_payoff}}
\end{figure}

As a result, unlike in the previous section, here we cannot hope to solve the problem exactly. 
However, the next theorem shows that we can approximate the optimal utility with arbitrarily good precision.

\begin{theorem}\label{thm:strong-induc}
If $u^L$ is not max-degenerate, then for any $\varepsilon > 0$, the follower can strongly induce a profile $(\xx,j)$ that yields the optimal inducible utility up to an additive loss of at most $\varepsilon$. Furthermore, a matrix $\uff$ strongly inducing $(\xx,j)$ can be constructed in time polynomial in $\log(1/\varepsilon)$ (and the size of the representation of the game).
\end{theorem}

\begin{proof}
Let $(\xx^*,j)$ be a payoff-inducible profile that yields the optimal inducible payoff for the follower. By Theorem~\ref{thm:payoff-inducibility-M}, such a profile can be computed in polynomial time.

We begin by solving the following LP.
\begin{equation}\label{eq:LPstrong}
\begin{split}
\max_{\delta, \xx}  \quad& \delta \\
\text{s.t.} \quad 
& \xx \in \Delta^{m-1} \\
& u^F(\xx,j) \geq u^F(\xx^*,j) - \varepsilon \\
& u^L(\xx,j) = u^L(\xx^*,j) + \delta
\end{split}\end{equation}
Note that this LP can be solved in time polynomial in $\log(1/\varepsilon)$. 
Furthermore, note that the polytope of feasible points is not empty since $\delta=0$ and $\xx=\xx^*$ satisfy all the constraints. Finally, the LP is not unbounded since $\delta$ can be at most $\max_{i \in [m]} u^L(i,j) - u^L(\xx^*,j)$.

In the rest of this proof let $\delta$ and $\xx$ denote an optimal solution to this LP. Note that we can in particular assume that $\xx$ is a vertex of the convex polytope $P_\delta = \{\yy \in \Delta^{m-1} \, : \, u^L(\yy,j) = u^L(\xx^*,j) + \delta\}$. Indeed, given a solution $\delta, \xx$ to LP \eqref{eq:LPstrong}, if $\xx$ is not a vertex of $P_\delta$, then we consider the LP
\begin{align*}
\max_{\yy}  \quad& u^F(\yy,j)\\
\text{s.t.} \quad 
& \yy \in \Delta^{m-1} \\
& u^L(\yy,j) = u^L(\xx^*,j) + \delta
\end{align*}
It is known that a solution of an LP that is also a vertex of the feasible polytope can be computed in polynomial time~\citep{GLS1981ellipsoid}. Note that in this case the feasible polytope is exactly $P_\delta$. Let $\yy$ be an optimal solution that is a vertex of $P_\delta$. We know that $\xx \in P_\delta$ and $u^F(\xx,j) \geq u^F(\xx^*,j) - \varepsilon$, which implies that $u^F(\yy,j) \geq u^F(\xx^*,j) - \varepsilon$. But this means that $\delta, \yy$ is also an optimal solution to the original LP \eqref{eq:LPstrong}. Thus, by letting $\xx := \yy$, we indeed have that $\xx$ is a vertex of the convex polytope $P_\delta$.

Let us first handle the case where $\delta = 0$ by showing that $(\xx^*,j)$ itself can be strongly induced. Since $\delta=0$, it follows that $U_j(\xx^*) = \emptyset$. Indeed, if there exists $\hat{\yy} \in \Delta^{m-1}$ with $u^L(\hat{\yy},j) > u^L(\xx^*,j)$, then there exists $\yy$ on the segment $(\xx^*,\hat{\yy}]$ such that $u^F(\yy,j) \geq u^F(\xx^*,j) - \varepsilon$ (when $\yy$ is sufficiently close to $\xx^*$) and $u^L(\yy,j) > u^L(\xx^*,j)$, a contradiction to the optimality of $\delta=0$. Now, given that $U_j(\xx^*) = \emptyset$, we have that $u^L(\yy,j) \leq u^L(\xx^*,j)$ for all $\yy \in \Delta^{m-1}$. But since $u^L$ is not max-degenerate (in the sense of Definition~\ref{def:degenerate}), it follows that in fact $u^L(\yy,j) < u^L(\xx^*,j)$ for all $\yy \in \Delta^{m-1} \setminus \{\xx^*\}$. Thus, if the follower always best responds with strategy $j$, then $(\xx^*,j)$ will be the unique SSE. As seen before, it is easy to implement this behavior by reporting $\uff(i,j) = 1$ and $\uff(i,\ell)=0$ for all $i \in [m]$ and $\ell \in [n] \setminus \{j\}$.

In the rest of this proof, we consider the case $\delta > 0$ and show that $(\xx,j)$ can be strongly induced. Since $u^F(\xx,j) \geq u^F(\xx^*,j) - \varepsilon$, this means that at $(\xx,j)$ the follower achieves the optimal inducible utility up to an additive error of $\varepsilon$. Using the same notation as in the proof of Proposition~\ref{prp:key-condition-holds}, we let
\begin{equation*}
B = \{i \in [m] : \ee_i \cdot \xx = \beta_i\}
\end{equation*}
denote the set of boundary conditions of $\Delta^{m-1}$ that are tight for $\xx$. Note that since $\xx$ is a vertex of the polytope $P_\delta$, it follows that $B \neq \emptyset$. We let $\hh = \sum_{i \in B} \ee_i$. As in the proof of Proposition~\ref{prp:key-condition-holds}, we have that for all $\yy \in \Delta^{m-1}$ it holds that
\begin{equation}\label{eq:boundary-strong}
\hh \cdot (\yy - \xx) = \sum_{i \in B} \ee_i \cdot (\yy - \xx) \geq 0.
\end{equation}
Furthermore, since $\xx$ is a vertex of $P_\delta$, it follows that for all $\yy \in P_\delta \setminus \{\xx\}$ there exists $i \in B$ such that $\ee_i \cdot (\yy - \xx) > 0$, and thus
\begin{equation}\label{eq:boundary-polytope-strong}
\hh \cdot (\yy - \xx) > 0.
\end{equation}
Indeed, if $\ee_i \cdot (\yy - \xx) = 0$ for all $i \in B$ for some $\yy \in P_\delta \setminus \{\xx\}$, this would contradict the fact that $\xx$ is a vertex of $P_\delta$ (i.e.\ the unique point in $P_\delta$ for which the boundary conditions in $B$ are tight).

We are now ready to construct the payoff matrix reported by the follower. Pick an arbitrary $k \in \argmin_{\ell \in [n] \setminus \{j\}} u^L(\xx,\ell)$. For all $\yy \in \Delta^{m-1}$ let
\begin{equation}
\label{eq:util-strong}
\uff(\yy,\ell) = \begin{cases}
-u^L(\yy,\ell) & \text{if } \ell \in [n] \setminus \{j\}\\
-u^L(\yy,k) + \alpha \left(u^L(\xx,j) - u^L(\yy,j)\right) - \hh \cdot (\yy-\xx) & \text{if } \ell=j
\end{cases}
\end{equation}
where $\alpha = \left(2 \max_{i \in [m]} \max_{\ell \in [n]} \left|u^L(i,\ell)\right|  + m \right)/\delta > 0$. Note that we can compute the payoff matrix corresponding to this utility function in polynomial time. In the remainder of this proof, we show that $(\xx,j)$ is the unique SSE of the game $(u^L,\uff)$.

Clearly, $j$ is a best response at $\xx$, since 
\[
\uff(\xx,j) = -u^L(\xx,k) = - \min_{\ell \in [n] \setminus \{j\}} u^L(\xx,\ell) = \max_{\ell \in [n] \setminus \{j\}} \uff(\xx,\ell),
\]
by the choice of $k$.

Next, let us show that if $j$ is a best response at some $\yy \in \Delta^{m-1} \setminus \{\xx\}$, then $u^L(\yy,j) < u^L(\xx,j)$. Indeed, if $j$ is a best response at $\yy$, then in particular $\uff(\yy,j) \geq \uff(\yy,k)$, which implies that
\begin{equation}\label{eq:j-best-response}
\alpha \left(u^L(\xx,j) - u^L(\yy,j)\right) \geq \hh \cdot (\yy-\xx).
\end{equation}
Since $\hh \cdot (\yy-\xx) \geq 0$ by \eqref{eq:boundary-strong}, and $\alpha > 0$, it follows that $u^L(\xx,j) \geq u^L(\yy,j)$. It remains to show that $u^L(\xx,j) \neq u^L(\yy,j)$. But if $u^L(\xx,j) = u^L(\yy,j)$, then $\yy \in P_\delta \setminus \{\xx\}$ and so by \eqref{eq:boundary-polytope-strong} we have $\hh \cdot (\yy-\xx) > 0$, which contradicts \eqref{eq:j-best-response}.

Finally, it remains to show that if $\ell \in [n] \setminus \{j\}$ is a best response at some $\yy \in \Delta^{m-1}$, then it must be that $u^L(\yy,\ell) < u^L(\xx,j)$: Indeed, if $\ell \in [n] \setminus \{j\}$ is a best response at $\yy$, then in particular $\uff(\yy,j) \leq \uff(\yy,\ell)$, which by \eqref{eq:util-strong} means that
\begin{equation*}
\begin{split}
\alpha \left(u^L(\xx,j) - u^L(\yy,j)\right) &\leq -u^L(\yy,\ell) + u^L(\yy,k) + \hh \cdot (\yy - \xx)\\
&\leq -u^L(\yy,\ell) + u^L(\yy,k) + \|\hh\|_2  \|\yy - \xx\|_2\\
&\leq  2 \max_{i \in [m]} \max_{\ell' \in [n]} |u^L(i,\ell')| + \sqrt{m-1} \sqrt{m-1}\\
&\leq \alpha \delta
\end{split}
\end{equation*}
by the choice of $\alpha$. Thus, we obtain that $u^L(\xx,j) - u^L(\yy,j) \leq \delta$, which implies that $u^L(\yy,j) \geq u^L(\xx^*,j)$, i.e.\ $\yy \in \overline{U_j(\xx^*)}$ (since $U_j(\xx^*) \neq \emptyset$). Since $(\xx^*,j)$ is payoff-inducible, which means that $u^L(\xx^*,j) \geq M$, we can use Lemma~\ref{lmm:M-V} to obtain
\begin{equation*}
u^L(\xx,j) = u^L(\xx^*,j) + \delta > u^L(\xx^*,j) \geq \min_{\ell' \in [n] \setminus \{j\}} u^L(\yy,\ell') = u^L(\yy,\ell)
\end{equation*}
where the last equality comes from the fact that $\ell$ is a best response at $\yy$, i.e., in particular $\uff(\yy,\ell) = \max_{\ell' \in [n] \setminus \{j\}} \uff(\yy,\ell')$.
\end{proof}

\section{Directions for Future Work}\label{sec:future}

An interesting first question that emerges from our results, is how to design countermeasures to mitigate the potential loss of a learning leader, caused by possible deceptive behavior of the follower. This was considered in \citep{gan2019imitative}, where as a solution it was proposed that the leader could commit to a policy, which is a strategy conditioned on the report of the follower, instead of a strategy. However, in contrast to \citep{gan2019imitative}, where the follower's report is limited to a finite set of payoff matrices, computing the optimal policy in our model seems to be a very challenging problem.
In addition, it would be nice to explore whether the optimal follower payoff matrix (or a good approximation of it) can still be computed efficiently, when additional constraints on how much he can deviate from his true payoff matrix are imposed.
Finally, another interesting direction would be to quantify and provide tight bounds on the leader's utility loss, caused by the deceptive behavior of the follower.

\bibliographystyle{plainnat}
\bibliography{references.bib}

\end{document}